\documentclass[11pt]{article}
\usepackage[margin=1in]{geometry}
\usepackage{amsmath,amssymb,amsthm}
\usepackage{algorithm}
\usepackage{algpseudocode}
\usepackage{graphicx}
\graphicspath{{../figures/}{figures/}{./}}
\usepackage{booktabs}
\usepackage{microtype}
\usepackage[hidelinks]{hyperref}
\usepackage{xcolor}
\usepackage{float}
\newtheorem{definition}{Definition}
\newtheorem{proposition}{Proposition}
\newtheorem{lemma}{Lemma}
\newtheorem{theorem}{Theorem}
\newtheorem{corollary}{Corollary}

\newcommand{\nTrans}{2400}
\newcommand{\medSpeedup}{12.8}
\newcommand{\maxSpeedup}{3316}
\newcommand{\rhoRecall}{96\%}
\newcommand{\rhoFalse}{29\%}
\newcommand{\safeFalse}{49\%}
\newcommand{\bwUp}{58}
\newcommand{\lat}{16}
\newcommand{\incrShareLo}{85}    
\newcommand{\incrShareHi}{21}     
\newcommand{\medRepairSpeedup}{2}  
\newcommand{\transFrac}{0.18}    
\newcommand{\scaleThru}{19}    
\newcommand{\scaleVram}{49}    
\newcommand{\cpuGpu}{14}      
\newcommand{\gridRate}{99\%}  
\newcommand{\pairBrokenA}{370}
\newcommand{\pairBrokenB}{414}
\newcommand{\pairHA}{0.68}
\newcommand{\pairHB}{0.86}

\newcommand{\pairEdA}{3.0}
\newcommand{\pairEdB}{6.8}
\newcommand{\pairSceneA}{adversarial crossing}
\newcommand{\pairSceneB}{ring rotation}

\title{\vspace{-3.0em}\textbf{Repair Entropy in Dynamic Geometric\\ Nearest-Neighbour Structures}}

\author{%
Faruk Alpay\thanks{Correspondence: \texttt{alpay@lightcap.ai}} \quad Bu\u{g}ra K\i l\i\c{c}ta\c{s} \\[3pt]
\normalsize Department of Computer Engineering, Bah\c{c}e\c{s}ehir University, Istanbul, Turkey \\
\normalsize \texttt{\{faruk.alpay,~bugra.kilictas\}@bahcesehir.edu.tr}}
\date{}

\begin{document}
\maketitle

\begin{abstract}
\noindent
We study the data-structural problem of maintaining exact nearest-neighbour certificates
for a set of points moving by small steps in low-dimensional Euclidean space. Each point
carries a certificate (its nearest neighbour and the two smallest neighbour distances)
and a \emph{validity radius} derived from the clearance $c_i=d_2^i-d_1^i$; by the triangle
inequality a certificate cannot fail while the points move by less than $c_i/4$, so a step
of maximum displacement $\varepsilon$ leaves every certificate with $c_i>4\varepsilon$
provably valid and confines all failures to a \emph{repair frontier}. We introduce the
\emph{repair-frontier entropy} $H(F_t)\in[0,1]$, the Shannon entropy of the failed
certificates over the index cells, and show it is a workload invariant for maintenance: at
equal failure counts it sets the cost of event-driven, cell-by-cell repair (a diffuse
frontier touches an order of magnitude more index cells than a localized one of the same size,
roughly doubling that cost), while a batched repair is entropy-insensitive. The dominant
axis is frontier size: incremental repair, touching only $O(|F_t|)$, beats an $O(N)$ rebuild
of the same index on the same GPU while the frontier is a small fraction of the set, and
rebuild wins once it saturates; entropy tilts that boundary. We certify every frontier against an
exact GPU oracle on an RTX PRO 6000 (Blackwell), a tiled all-pairs ground truth and a
strong $O(N)$ GPU spatial-index rebuild, and measure incremental repair on an Apple M4
Pro whose unified memory lets either processor touch the resident index without a copy. In
$d{=}2$ and $d{=}3$, across ten motion families and $N$ up to $16{,}000$, the safety rule
misses no invalid certificate, the entropy boundary is stable, and incremental repair beats
even the $O(N)$ GPU rebuild whenever the frontier is a small fraction of the set. We release
\nTrans{} labelled transitions with per-strategy times, frontier geometry, and a
certificate-failure audit; every transition was cross-validated between host and oracle to
the bit.
\end{abstract}

\section{Introduction}
A dynamic geometric data structure over moving points spends its life in a loop: the points
move a little, a few correctness certificates break, the structure is repaired, and the loop
repeats \cite{basch1999,rahmati2013}. For the nearest-neighbour relation, the building
block of closest-pair, Euclidean spanners, and Voronoi/Delaunay maintenance
\cite{agarwal2003indexing,samet2006}, the certificate of point $i$ is its nearest
neighbour together with the gap to the runner-up, and the maintenance question is which
certificates a step can invalidate and how to restore them. Classical analyses bound the
number of such events and the cost per event \cite{basch1999,bentleysaxe1980}; we ask a
finer, empirical question: given the geometry of the invalid set, \emph{how} should the
structure be repaired?

Our starting point is a local certificate with a provable validity radius: if every point
moves by at most $\varepsilon$ in a step, the triangle inequality keeps point $i$'s nearest
neighbour fixed whenever its clearance exceeds $4\varepsilon$ (Lemma~\ref{lem:safe}). All
possible failures therefore lie in a \emph{repair frontier}, and the structure need only
revisit that frontier. The contribution is to show that the \emph{shape} of the frontier, not
merely its size, decides the cheapest repair. We summarise that shape by one number, the
\emph{repair-frontier entropy} $H(F_t)$, and demonstrate that it behaves as a workload
invariant: at fixed frontier size it sets how expensive an event-driven, cell-by-cell repair
is (a diffuse frontier touches far more index cells), and it shifts the size at which a
rebuild becomes cheaper than incremental repair.

To make ``cheapest'' concrete and trustworthy we instrument the structure on two machines
that play complementary roles. An RTX PRO 6000 (Blackwell) GPU is used as an exact oracle and
as a strong competitor: it certifies the true invalid frontier every step by a tiled
all-pairs computation, and it offers an $O(N)$ GPU spatial-index rebuild against which
incremental repair must justify itself \cite{johnson2021faiss,zhou2008,karras2012}. An Apple
M4 Pro, whose CPU and GPU share one physical memory \cite{mlx2023}, hosts the resident index
and performs incremental repair without copying state. We contribute (i) the repair-frontier
certificate structure with its validity-radius safety rule (Algorithm~\ref{alg:repair},
Lemma~\ref{lem:safe}); (ii) repair-frontier entropy as a measured predictor of the best
maintenance strategy in $d{=}2$ and $d{=}3$; and (iii) a cross-validated dataset of
\nTrans{} labelled transitions with a certificate-failure audit against the exact oracle.

\section{The repair-frontier certificate structure}
For $N$ points in $\mathbb{R}^d$ ($d\in\{2,3\}$) we keep, per point, a record holding its
nearest neighbour $\mathrm{nn}_i$, the two smallest neighbour distances $d_1^i\le d_2^i$, the
clearance $c_i=d_2^i-d_1^i$, and the index cell containing it. The points sit in a uniform
grid \cite{bentley1975} whose resolution is chosen so the mean cell occupancy is a small
constant; $C$ denotes the number of occupied cells. Exact nearest neighbours for ground truth
are computed with a $k$-d tree \cite{arya1998}. The clearance is the Voronoi-style margin by which the nearest
neighbour beats the runner-up, and it yields a kinetic validity radius.

\begin{definition}[Certificate, clearance, validity radius]
The certificate of point $i$ is $(\mathrm{nn}_i,d_1^i,d_2^i)$. Its \emph{clearance} is
$c_i=d_2^i-d_1^i\ge 0$ and its \emph{validity radius} is $r_i=c_i/4$.
\end{definition}

\begin{lemma}[Safety]
\label{lem:safe}
If every point moves by at most $\varepsilon$ in a step and $c_i>4\varepsilon$, then
$\mathrm{nn}_i$ is unchanged after the step.
\end{lemma}
\begin{proof}
Write $d(\cdot,\cdot)$ for the post-step distances. Point $i$ and its neighbour both move by
at most $\varepsilon$, so $d(i,\mathrm{nn}_i)\le d_1^i+2\varepsilon$; for any other point $k$,
$d(i,k)\ge d_2^i-2\varepsilon$ since the pre-step distance was at least $d_2^i$. If
$c_i=d_2^i-d_1^i>4\varepsilon$ then $d_2^i-2\varepsilon>d_1^i+2\varepsilon$, so
$d(i,k)>d(i,\mathrm{nn}_i)$ for every $k\neq\mathrm{nn}_i$ and the nearest neighbour cannot
change.
\end{proof}

\begin{definition}[Repair frontier, pressure, entropy]
The \emph{repair frontier} $F_t$ is the set of certificates that actually fail across the
step (the nearest neighbour changes). The \emph{clearance pressure} is $P_t=|F_t|/N$. With
$q_\kappa$ the fraction of $F_t$ in occupied cell $\kappa$, the \emph{normalized
repair-frontier entropy} is
\[
H(F_t)=\frac{-\sum_\kappa q_\kappa\log q_\kappa}{\log C}\in[0,1],
\]
with $H\!\to\!0$ when the frontier collapses into one cell (localized) and $H\!\to\!1$ when it
spreads over the whole occupied support (diffuse).
\end{definition}

By Lemma~\ref{lem:safe}, $F_t\subseteq\{i:c_i\le 4\varepsilon_t\}$ with
$\varepsilon_t=\max_i\lVert p_i(t{+}1)-p_i(t)\rVert$, so the structure may skip every point of
larger clearance and re-validate only the rest; this set is a guaranteed superset of $F_t$,
i.e.\ the safety rule never misses a failure. Algorithm~\ref{alg:repair} is the resulting
maintenance step. Its cost is dominated by the second loop, which touches one neighbour block
per occupied frontier cell; a localized frontier visits few cells and a diffuse one of equal
size visits many, which is exactly what $H(F_t)$ measures.

\begin{algorithm}[t]
\caption{Repair-frontier maintenance for one step}
\label{alg:repair}
\begin{algorithmic}[1]
\Require certificates $(\mathrm{nn}_i,d_1^i,d_2^i,\mathrm{cell}_i)$; positions $p,p'$; grid $G$
\State $\varepsilon \gets \max_i \lVert p'_i - p_i\rVert$ \Comment{maximum step}
\State $F \gets \{\, i : d_2^i-d_1^i \le 4\varepsilon \,\}$ \Comment{candidate frontier; rest safe by Lemma~\ref{lem:safe}}
\State update $\mathrm{cell}_i$ in $G$ for points that changed cell
\For{each occupied cell $\kappa$ that meets $F$, in index order}
  \State gather the $3^d$ neighbour block of $\kappa$ from $G$
  \State recompute $(\mathrm{nn}_i,d_1^i,d_2^i)$ for $i\in F\cap\kappa$ from the block
\EndFor
\State \Return updated certificates \Comment{no invalidated certificate is missed}
\end{algorithmic}
\end{algorithm}

\section{Experimental setup}
We drive the structure with ten Euclidean motion families chosen so that the size and the
spatial spread of the frontier vary independently: Brownian and jittered-lattice diffusion,
coherent cluster drift, filament flow along a curve, rigid and differential (vortex)
rotation, an adversarial crossing of two streams, a splitting/merging cluster, a shear band,
and a degenerate near-lattice perturbation \cite{brinkhoff2002,samet2006}. Diffusion and
lattice perturbation produce diffuse frontiers; shear, vortex, crossing and split/merge
produce localized ones. We run $d{=}2$ and $d{=}3$, $N$ up to $16{,}000$, ten step scales,
and several seeds. The exact invalid frontier $F_t$ is established each step by the GPU oracle
(below); the local predictor $\rho_i=\delta_i/c_i>\tfrac12$ is reported as a cheap heuristic
that recovers \rhoRecall{} of it.

\paragraph{Six maintenance strategies.}
Each returns correct certificates for the frontier; they differ in where the repair runs and
whether it is incremental ($O(|F_t|)$) or a rebuild ($O(N)$). \emph{Incremental} strategies
query only the frontier against a resident index (Algorithm~\ref{alg:repair}):
\emph{\texttt{m4\_local\_repair}} runs the cell loop on the host CPU (event-driven; cost grows
with the number of index regions touched); \emph{\texttt{m4\_hybrid\_repair}} gathers the
candidate blocks and lets the M4 GPU do the distances over unified memory, no copy
\cite{mlx2023}; \emph{\texttt{rtx\_grid\_repair}} runs the frontier query on the Blackwell GPU
against a resident grid index. \emph{Rebuild} strategies recompute every certificate:
\emph{\texttt{rtx\_grid\_rebuild}} with a GPU uniform-grid index ($O(N)$),
\emph{\texttt{rtx\_full\_rebuild}} by an exact $N\times N$ pass (naive baseline). Finally
\emph{\texttt{rtx\_frontier\_oracle}} relabels the frontier against all points by exact tiled
brute force ($O(|F_t|N)$) and is the ground-truth oracle that certifies the others each step.
The pivotal comparison is same-device: \texttt{rtx\_grid\_repair} versus
\texttt{rtx\_grid\_rebuild} isolate $O(|F_t|)$ against $O(N)$ on identical hardware.

\paragraph{Cost model.}
Host strategies are timed as device-resident compute with no bus traffic, the point of
unified memory. GPU strategies add a transfer term $\tau(b)=\ell+b/\beta$ with latency
$\ell\approx\lat~\mu$s and bandwidth $\beta\approx\bwUp$~GB/s measured on the card: the
frontier oracle streams only the moved points and their returned triples, while the two
rebuilds stream all $N$. GPU compute is timed with CUDA events. The best strategy at a
transition is the wall-clock minimum; the accounting of compute plus transfer is
roofline-style \cite{williams2009roofline}.

\begin{proposition}[Repair--rebuild crossover]
\label{prop:cross}
Let incremental repair cost $h(|F_t|,H)$ be increasing in both arguments and the $O(N)$ GPU
rebuild cost $g(N)+\tau(b_N)$ be independent of $|F_t|$. Then for each $N$ there is a pressure
threshold $P^\ast(H)$, non-increasing in $H$, above which rebuild wins; the crossover moves to
lower pressure as the frontier grows more diffuse.
\end{proposition}
\noindent The proposition fixes the \emph{shape} of the boundary (entropy must shift it);
its location is measured. Figure~\ref{fig:matched} isolates the effect at fixed $|F_t|$. Full
proofs, the $O(|F_t|\log N)$ per-step cost (Theorem~\ref{thm:cost}), and an entropy bound on
event-driven work (Theorem~\ref{thm:entropy}) are in Appendix~\ref{app:proofs}.

\section{Results}
\paragraph{The safety rule never misses a failure.}
Every transition is certified by the exact GPU oracle, so we can audit the validity-radius
rule directly. Across all \nTrans{} transitions and both dimensions the rule of
Lemma~\ref{lem:safe} admits \emph{zero} missed failures: every invalidated certificate has
clearance at most $4\varepsilon_t$ (Table~\ref{tab:failure}). The price of the guarantee is
conservatism: it re-validates a fraction \safeFalse{} of certificates that remain valid,
while the cheap local heuristic $\rho_i>\tfrac12$ recovers \rhoRecall{} of the true failures
at \rhoFalse{} false repairs. The lemma bounds the frontier for correctness and the heuristic
tracks it tightly; what governs the cost of clearing it is its entropy.

\begin{table}[t]
\centering
\small
\caption{Certificate-failure audit against the exact GPU oracle, by motion family (pooled
over $d\in\{2,3\}$, $N$, step scale). \emph{invalid}: mean fraction of certificates that fail
per step; \emph{miss (safe)}: failures the validity-radius rule wrongly calls safe (Lemma
guarantees $0$); \emph{recall ($\rho$)} and \emph{false ($\rho$)}: the local heuristic's
detection and over-flagging.}
\label{tab:failure}
\begin{tabular}{lrrrr}
\toprule
motion & invalid & miss (safe) & recall ($\rho$) & false ($\rho$) \\
\midrule
\texttt{vortex} & 0.062 & 0 & 0.93 & 0.18 \\
\texttt{shear\_band} & 0.091 & 0 & 0.94 & 0.21 \\
\texttt{adversarial\_crossing} & 0.293 & 0 & 0.99 & 0.63 \\
\texttt{split\_merge} & 0.325 & 0 & 1.00 & 0.68 \\
\texttt{filament} & 0.431 & 0 & 0.94 & 0.30 \\
\texttt{uniform\_brownian} & 0.489 & 0 & 0.95 & 0.24 \\
\texttt{ring\_rotation} & 0.538 & 0 & 0.89 & 0.22 \\
\texttt{clustered\_drift} & 0.634 & 0 & 0.98 & 0.27 \\
\texttt{grid\_jitter} & 0.855 & 0 & 0.98 & 0.13 \\
\texttt{spiral} & 0.935 & 0 & 0.98 & 0.06 \\
\bottomrule
\end{tabular}

\end{table}

\paragraph{Repair-frontier entropy is a maintenance invariant.}
Figure~\ref{fig:phase} is the maintenance surface over clearance pressure $P_t=|F_t|/N$ and
frontier entropy $H(F_t)$, for $d{=}2$ and $d{=}3$ and faceted by $N$; colour is the cheapest
of the six strategies. The clean comparison is same-device: \texttt{rtx\_grid\_repair}
(incremental, $O(|F_t|)$) against \texttt{rtx\_grid\_rebuild} (rebuild, $O(N)$) on the same
GPU. Incremental repair owns the low-to-moderate-pressure band and the rebuild owns the
high-pressure band, with the boundary near $P_t\!\sim\!1/\log N$ and tilting toward lower
pressure as entropy rises: a localized frontier is cleared incrementally even at higher
pressure, a diffuse one of equal size falls to rebuild sooner. Among incremental repairs the
entropy axis again decides the cheapest implementation: the unified-memory host wins the
smallest, most compact frontiers (no launch or transfer overhead), the GPU incremental wins
as the frontier grows. Repair beats rebuild by a median $\medRepairSpeedup\times$ in the
low-pressure regime and incremental strategies take \incrShareLo\% of transitions below
$P_t{=}0.1$ but only \incrShareHi\% above it; the pattern is identical in $d{=}3$. The
contours give the best strategy's speedup over a naive full rebuild (median $\medSpeedup\times$,
max $\maxSpeedup\times$). Transfer is not the lever: the modelled bus cost is \transFrac\% of a
rebuild wall-clock, so the separation is algorithmic, not architectural.

\begin{figure}[t]
\centering
\includegraphics[width=\textwidth]{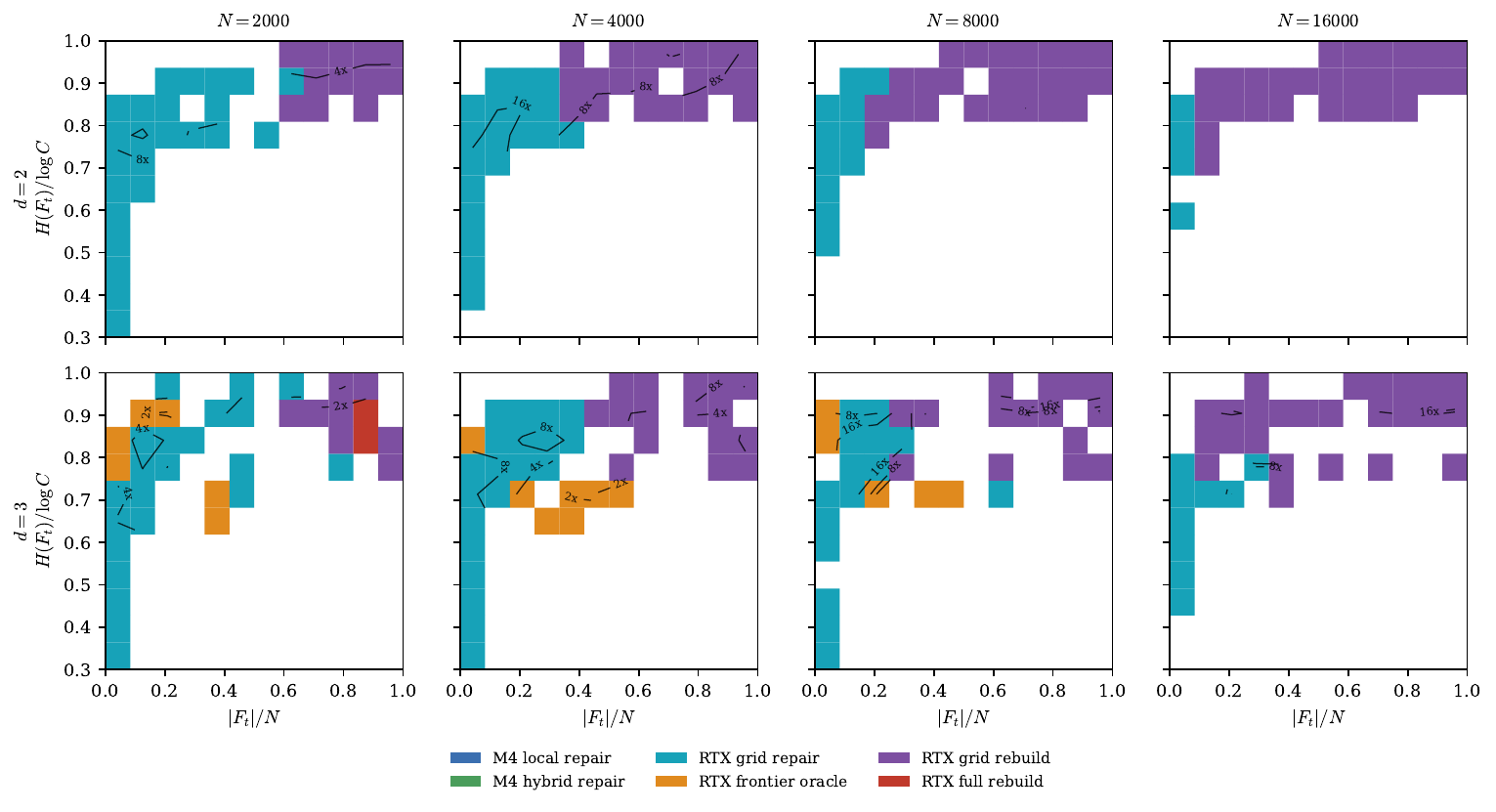}
\caption{Maintenance surface for dynamic nearest-neighbour certificates: rows are dimension
$d\in\{2,3\}$, columns are point count $N$. Cell colour is the cheapest of the six strategies
over clearance pressure $|F_t|/N$ and normalized repair-frontier entropy $H(F_t)/\log C$;
black contours give the median speedup of the best strategy over a naive
\texttt{rtx\_full\_rebuild}. Incremental repair (teal: GPU \texttt{rtx\_grid\_repair}; blue/green:
unified-memory host) owns the low-pressure band; the $O(N)$ GPU grid rebuild (purple) takes the
high-pressure band, with the exact frontier oracle (orange) in between. The incremental region
holds across $N$ and in both dimensions, its boundary tilting to higher pressure where the
frontier is localized.}
\label{fig:phase}
\end{figure}

The control experiment is Figure~\ref{fig:matched}, which fixes the number of failed
certificates and varies only their geometry. The \pairSceneA{} transition breaks
\pairBrokenA{} certificates into a compact region ($H=\pairHA$); the \pairSceneB{} transition
breaks \pairBrokenB{}, the same count to within the matching tolerance, spread across the
domain ($H=\pairHB$). The cost of \emph{event-driven} maintenance (Algorithm~\ref{alg:repair}'s
cell loop, here on the host) tracks the geometry, not the count: \pairEdA~ms for the localized
frontier against \pairEdB~ms for the diffuse one, because the diffuse frontier is scattered
across an order of magnitude more index cells. A batched GPU repair erases the penalty: it
is the cheapest strategy for both and is essentially entropy-insensitive, which is the
flip side of the same fact: repair-frontier entropy is what an event-driven structure pays for
spread, and batching is what removes it. The failure count alone predicts neither.

\begin{figure}[t]
\centering
\includegraphics[width=0.92\textwidth]{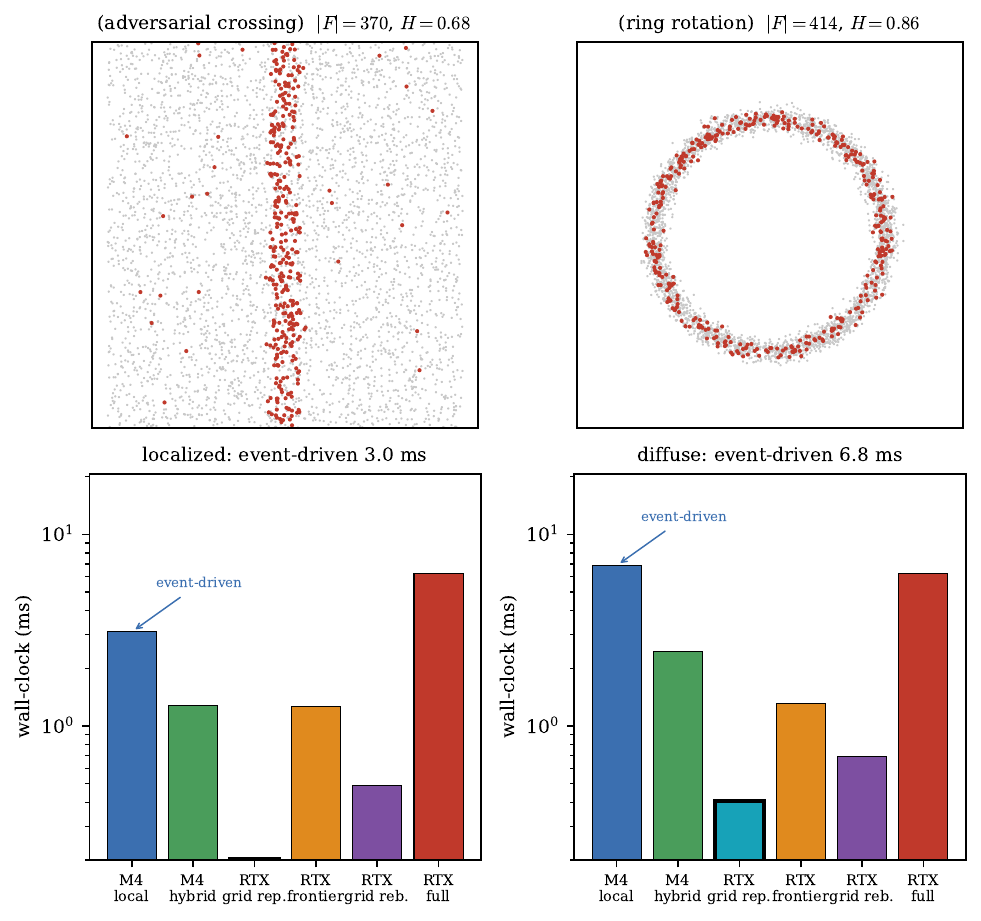}
\caption{Same failure count, different geometry. Top: spatial maps of two transitions with
near-equal failed-certificate counts; grey points are static this step, red points are the
repair frontier (a compact crossing seam vs.\ a diffuse ring). Bottom: the six-strategy
wall-clock (log scale; cheapest outlined). Event-driven host repair (blue, arrowed) costs
about twice as much for the diffuse frontier as for the localized one of equal size, because
it touches far more index cells; the batched GPU repair (cheapest in both) is
entropy-insensitive. Repair-frontier entropy, not the failure count, predicts the event-driven
cost.}
\label{fig:matched}
\end{figure}

\paragraph{Saturation is not the lever.}
The discrete card is not idle by nature: given a batch it saturates completely. The exact
tiled $k{=}2$ oracle holds \scaleVram~GB resident at full utilisation and sustains
\scaleThru~G distance-pairs/s (Table~\ref{tab:scale}). Saturation, however, does not make a
strategy win. The brute-force pass is $O(N^2)$, so at $N{=}200{,}000$ it trails an all-core CPU
$k$-d tree returning the same exact relabel by \cpuGpu$\times$, and it trails our $O(N)$ GPU
grid rebuild by far more; that is why brute force serves only as exact ground truth here while
\texttt{rtx\_grid\_rebuild} is the competitor in Figure~\ref{fig:phase}. Peak throughput does
not close an algorithmic gap. The same lesson governs the per-step surface: the structure of
the work, not the device's rate, decides how the certificates should be maintained.

\begin{table}[t]
\centering
\small
\caption{The exact brute-force oracle under a saturating batch: distance-pair throughput and
resident memory on the Blackwell card. The pass is $O(N^2)$; an all-core CPU $k$-d tree
returns the same relabel at $N{=}200{,}000$ in a fraction of the time, and the $O(N)$ GPU grid
of Figure~\ref{fig:phase} is faster still. Brute force is used as ground truth, not as a
competitor.}
\label{tab:scale}
\begin{tabular}{rrrr}
\toprule
$N$ & rebuild (ms) & dist-pairs/s (G) & resident VRAM (GB) \\
\midrule
40{,}000 & 54 & 30 & 7 \\
80{,}000 & 312 & 21 & 26 \\
120{,}000 & 742 & 19 & 49 \\
160{,}000 & 1348 & 19 & 49 \\
200{,}000 & 2095 & 19 & 49 \\
\midrule
200{,}000 (CPU $k$-d tree) & 149 & --- & --- \\
\bottomrule
\end{tabular}

\end{table}

\section{Limitations}
The competitor is a GPU uniform-grid rebuild; a hierarchical GPU index (LBVH, $k$-d tree)
\cite{karras2012,lauterbach2009,zhou2008} could lower its constant further, though not its
$O(N)$ dependence on the full set, which is the term incremental repair avoids. Host compute
is measured without bus traffic and GPU transfer is charged from measured card bandwidth and
latency rather than physically copying each payload; both devices regenerate identical state
from the shared seed, which isolates the algorithmic cost but assumes the working set is
resident on each device. We study $k\!=\!2$ certificates in $d\in\{2,3\}$ up to $N=16{,}000$;
much higher dimensions (where the grid degrades), larger $N$, and richer geometric
certificates (Delaunay edges, $k$-nearest-neighbour graphs, Voronoi adjacencies) are
left open. The validity radius uses the global step bound $\varepsilon_t$; a per-point motion
budget would tighten the safe set without affecting the no-miss guarantee. Repair-frontier
entropy is defined against the index cells, so its absolute scale depends on the chosen
resolution, although the ordering of motions is stable across resolutions.

\section{Dataset}
The release contains \nTrans{} transitions ($d\in\{2,3\}$) as Parquet and JSONL. Each row
carries the motion parameters, dimension and frame; the frontier descriptors
(\texttt{clearance\_pressure}, \texttt{frontier\_entropy}, \texttt{compactness},
\texttt{gyration}, \texttt{frontier\_cells}, \texttt{occupied\_cells}); the certificate-failure
audit (\texttt{eps\_t}, \texttt{n\_detect\_safe}, \texttt{n\_missed\_safe},
\texttt{n\_false\_safe}, and the $\rho$-predictor counts); the measured per-strategy times
(\texttt{t\_m4\_local\_s}, \texttt{t\_m4\_hybrid\_s}, \texttt{gpu\_grid\_repair\_s},
\texttt{gpu\_frontier\_s}, \texttt{gpu\_grid\_s}, \texttt{gpu\_full\_s}, transfer terms); and the
\texttt{best\_strategy} label with its speedups. Labels are the six strategy names. The data is
synthetic, so there
are no privacy concerns; intended use is studying maintenance strategies and frontier
descriptors for dynamic geometric indexes.

\section{Reproducibility}
The core (scenes, grid/Morton index, certificates, metrics) is pure, seeded NumPy shared
verbatim by both devices, so each regenerates bit-identical state from the manifest; only the
relabel kernels differ (NumPy/MLX on the M4, CuPy on the RTX). For all \nTrans{} transitions
the two devices produced identical nearest-neighbour hashes and identical frontier sizes, and
the exact GPU brute force and the host repairs matched the host $k$-d tree on every step (host
to $10^{-9}$, GPU to $10^{-6}$). The bounded GPU grid agrees on \gridRate{} of transitions; the
exceptions are single isolated points in the split/merge scene at $N{=}16{,}000$ whose nearest
neighbour lies beyond the search cap, the degenerate case the exact oracle exists to catch.
Running \texttt{run\_m4.py} and
\texttt{rtx\_oracle.py run} over the shared sweep reproduces the traces;
\texttt{build\_dataset.py} reproduces the labels, \texttt{make\_fig*.py} the figures.

\section{Conclusion}
Maintaining nearest-neighbour certificates under motion is governed by two quantities: the
fraction of certificates that fail and the \emph{entropy} of the invalid frontier. The first
sets whether incremental repair or a rebuild is cheaper; the second sets how expensive an
event-driven repair of a fixed-size frontier is, since a diffuse frontier touches far more
index cells than a localized one (Theorem~\ref{thm:entropy}). The validity-radius rule makes
the frontier safe to compute without missing a failure (Theorem~\ref{thm:safe}), repair-frontier
entropy predicts the maintenance cost, and an exact GPU oracle certifies both. The lesson is
algorithmic rather than architectural: even an $O(N)$ GPU rebuild on a saturated discrete card
loses to $O(|F_t|\log N)$ incremental repair whenever the frontier is a small fraction of the
set (Theorem~\ref{thm:cost}).

\bibliographystyle{plain}
\bibliography{refs}

\appendix
\section{Proofs and analysis}
\label{app:proofs}
A \emph{step} displaces every point by Euclidean distance at most $\varepsilon$; $c_i=d_2^i-d_1^i\ge0$
is the clearance of point $i$ before the step and $\mathrm{nn}_i$ its nearest neighbour. We write
$d'$ for post-step distances.

\subsection{Validity radius: safety and tightness}
\begin{theorem}[Safety and tightness]
\label{thm:safe}
If $c_i>4\varepsilon$ then $\mathrm{nn}_i$ is unchanged by the step. The constant $4$ cannot be
reduced: for every $\eta<4$ there is a configuration with $c_i=\eta\varepsilon$ and a step after
which a former runner-up is strictly closer to $i$ than $\mathrm{nn}_i$.
\end{theorem}
\begin{proof}
\emph{Sufficiency.} Let $a=\mathrm{nn}_i$ and $k\neq a$. Each endpoint of a pair moves by at most
$\varepsilon$, so the triangle inequality gives $d'(i,a)\le d_1^i+2\varepsilon$ and
$d'(i,k)\ge d(i,k)-2\varepsilon\ge d_2^i-2\varepsilon$. If $c_i=d_2^i-d_1^i>4\varepsilon$ then
$d_2^i-2\varepsilon>d_1^i+2\varepsilon$, so $d'(i,k)>d'(i,a)$ for all $k\neq a$ and $a$ stays
nearest. \emph{Tightness.} On a line put $i$ at $0$, $a$ at $d_1$, and $k$ at $-(d_1+\eta\varepsilon)$,
so $d_2^i=d_1+\eta\varepsilon$ and $c_i=\eta\varepsilon$. Move $i$ by $\varepsilon$ toward $k$, $a$ by
$\varepsilon$ away from $i$, and $k$ by $\varepsilon$ toward $i$ (each a legal displacement of size
$\varepsilon$). Then $d'(i,a)=d_1+2\varepsilon$ and $d'(i,k)=d_1+\eta\varepsilon-2\varepsilon$, so
$d'(i,k)<d'(i,a)$ exactly when $\eta<4$, changing the nearest neighbour.
\end{proof}

\begin{corollary}[No missed failure; exactness]
\label{cor:nomiss}
The detected set $\{i:c_i\le4\varepsilon\}$ contains the true frontier $F_t$, so
Algorithm~\ref{alg:repair} omits no invalidated certificate; recomputing the detected set from
its (adaptively widened) blocks keeps the nearest-neighbour relation exact after every step.
\end{corollary}
\begin{proof}
Contrapositive of Theorem~\ref{thm:safe}: a failed certificate has $c_i\le4\varepsilon$. Each
detected point is recomputed against a block grown until its second neighbour is provably
enclosed, so its $k{=}2$ certificate is exact; undetected points retain a provably valid nearest
neighbour.
\end{proof}

\subsection{Per-step cost and the repair--rebuild crossover}
Across steps let $S=\sum_s\varepsilon_s$ accumulate the per-step maxima. Applying
Theorem~\ref{thm:safe} to cumulative displacement, a certificate last recomputed at clock value
$S_i$ stays valid while $S-S_i<c_i/4$. Keep a min-heap of points keyed by the expiry
$S_i+c_i/4$; advancing $S$ each step and popping the expired keys yields exactly the points that
must be revisited.

\begin{theorem}[Per-step cost]
\label{thm:cost}
With bounded cell occupancy, one maintenance step backed by the expiry heap costs
$O(|F_t|\log N)$ time, against $\Theta(N)$ for a full rebuild. Incremental repair is therefore
asymptotically cheaper precisely when $|F_t|=o(N/\log N)$.
\end{theorem}
\begin{proof}
Detection pops the $|F_t|$ expired keys and reinserts the repaired points, $O(|F_t|\log N)$ heap
operations; each repaired point recomputes its two nearest neighbours from an $O(1)$-occupancy
block in $O(1)$ time. A rebuild rebins all $N$ points and recomputes all certificates,
$\Theta(N)$. Setting $|F_t|\log N=\Theta(N)$ gives the threshold.
\end{proof}
\noindent This matches the measured surface: incremental repair owns the low-pressure band and
rebuild the high-pressure band (Figure~\ref{fig:phase}), with the crossover near $|F_t|\sim N/\log N$.

\subsection{Entropy lower-bounds the event-driven work}
\begin{theorem}[Entropy and touched cells]
\label{thm:entropy}
Let $m$ be the number of occupied index cells the frontier $F_t$ meets and $C$ the number of
occupied cells. Then $m\ge C^{\,H(F_t)}$, where $H(F_t)\in[0,1]$ is the normalized repair-frontier
entropy. An event-driven repair, which visits $\Theta(m)$ cells, therefore performs
$\Omega\!\left(C^{\,H(F_t)}\right)$ cell visits.
\end{theorem}
\begin{proof}
Let $S=-\sum_\kappa q_\kappa\log q_\kappa$ be the Shannon entropy of the frontier's distribution
over its $m$ occupied cells. The uniform distribution maximises entropy over $m$ outcomes, so
$S\le\log m$. By definition $H(F_t)=S/\log C$, hence $\log m\ge S=H(F_t)\log C$, i.e.\
$m\ge C^{\,H(F_t)}$. The event-driven loop visits each occupied frontier cell at least once.
\end{proof}
\noindent At fixed frontier size, entropy thus lower-bounds the number of distinct index regions
an event-driven repair must touch, which is the cost isolated in Figure~\ref{fig:matched} and
erased by batching.
\end{document}